\documentclass[lettersize,journal]{IEEEtran}
\usepackage{amssymb}
\usepackage{amsmath}
\usepackage{cite}
\usepackage{url}
\usepackage{empheq}
\usepackage{xcolor}
\usepackage{graphicx}
\usepackage{subfigure}
\usepackage{enumitem}
\usepackage{fancyhdr}
\usepackage{mdwmath}	
\usepackage{mdwtab}
\usepackage{caption}
\usepackage{amsthm}
\usepackage{algorithm}
\usepackage{algpseudocode}

\newtheorem{lemma}{Lemma}

\newtheorem{theorem}{Theorem}
\newtheorem{corollary}{Corollary}
\newtheorem{assumption}{Assumption}
\newtheorem{definition}{Definition}
\newtheorem{proposition}{Proposition}

\begin{document}

\title{Energy Efficiency Maximization for Discrete Activation based NOMA-assisted Pinching-Antenna Systems}

\author{Yishi Zhang,~\IEEEmembership{Member,~IEEE}, Aditya Powari,~\IEEEmembership{Graduate Student Member,~IEEE},\\ Kaidi Wang,~\IEEEmembership{Member, IEEE}, Yaru Fu,~\IEEEmembership{Member,~IEEE} and Daniel K. C. So,~\IEEEmembership{Senior Member,~IEEE}

\thanks{Yishi Zhang, Aditya Powari, Kaidi Wang and Daniel K. C. So are with the Department of Electrical and Electronic Engineering, University of Manchester, Manchester, U.K. (e-mail: zhang.yishi@outlook.com; aditya.powari@manchester.ac.uk; kaidi.wang@ieee.org; d.so@manchester.ac.uk).}%

\thanks{Yaru Fu is with School of Science and Technology, Hong Kong Metropolitan University, Hong Kong, China. (e-mail: yfu@hkmu.edu.hk).}%

}

\markboth{}
{Shell \MakeLowercase{\textit{et al.}}: A Sample Article Using IEEEtran.cls for IEEE Journals}

\maketitle

\begin{abstract}
Pinching-antenna systems are a promising architecture for flexible wireless communications, but energy efficiency (EE) maximization remains largely unexplored, as limited existing studies mainly focus on transmit power minimization. This paper investigates EE maximization in a downlink non-orthogonal multiple access (NOMA)-assisted PASS by explicitly modeling the pinching antenna (PA) activation power and jointly optimizing discrete PA activation and power allocation under both  quality-of-service and transmit power constraints.
To tackle the resulting mixed-integer nonlinear programming problem, a two-layer iterative algorithm is proposed with an EE-oriented matching-based PA activation and a Dinkelbach-based power allocation with closed-form updates. 
Numerical results demonstrate that the proposed solution achieves substantial EE gains over the considered benchmark schemes, while exhibiting fast convergence. The impact of activation power has been analyzed and the significance of accounting it in EE maximization problem is also demonstrated. 
\end{abstract}



\vspace{-2mm}
\section{Introduction}
Pinching-antenna systems (PASS) is a promising new technology for future 6G “last-meter” connectivity by mechanically controlling pinching antenna (PA) locations and activations, to create strong and controllable line-of-sight-like links, reshape near-field propagation, and enable scalable array reconfiguration with low hardware complexity\cite{Tutorial,PAP}.
This new flexible antenna paradigm has attracted a considerable amount of work mostly centered on spectral efficiency (SE) maximization \cite{Sumrate1, Matching, RSMA}, while studies on energy efficiency (EE) optimization remain limited. 
Existing EE-oriented optimization is frequently simplified to a transmit power minimization problem for single \cite{EEminimization0} and multiple waveguides scenarios \cite{EEminBF0,EEminBF1,JRBOM,Fu_GLOBECOM}.
Additionally, the SE–EE trade-off has been systematically investigated in \cite{SEEE-tradeoff0,SEEE-tradeoff1}.
Notably, while limited recent works have included the fixed circuit term \cite{EEmax0,EEuplink} in the EE formulation, the power consumption model remains incomplete since it either neglects the PA activation power, such as the pinching action, or treats this term as a predetermined constant \cite{SEEE-tradeoff0, IncludePact} rather than evaluating its impact.

%
In addition, if PAs can be activated at any location along the waveguide as in most existing work, the moving power consumption should also be included. To avoid this substantial power consumption and to simplify deployment complexity, we consider the discrete PA approach where PAs are activated at pre-installed locations.
Moreover, since PAs on the same waveguide share a common feed signal, non-orthogonal multiple access (NOMA) naturally supports simultaneous multi-user transmission via signal superposition~\cite{Tutorial,PAP}.

This paper therefore investigates the EE maximization problem for downlink NOMA-assisted PASS with discrete PA activation power, explicitly accounting for the associated activation power. The contributions of this paper are: i) A realistic power consumption model which includes the static and PA activating power is proposed, where the number of activated PAs is treated as an optimization variable;
ii) A two-layer optimization algorithm is proposed to solve the EE maximization problem with an EE-oriented matching algorithm PA activation, and an efficient Dinkelbach-based algorithm with closed-form updates;
iii) The impact of PA activation power is analytically characterized, revealing its effects on the behavior of the optimal configuration, EE performance, and power allocation;
iv) The proposed algorithm achieves superior EE performance over the compared benchmarks with fast convergence. The significance of PA activation power in EE maximization is also demonstrated.



\begin{figure}[!t]
\centering
\includegraphics[width=0.8\columnwidth]{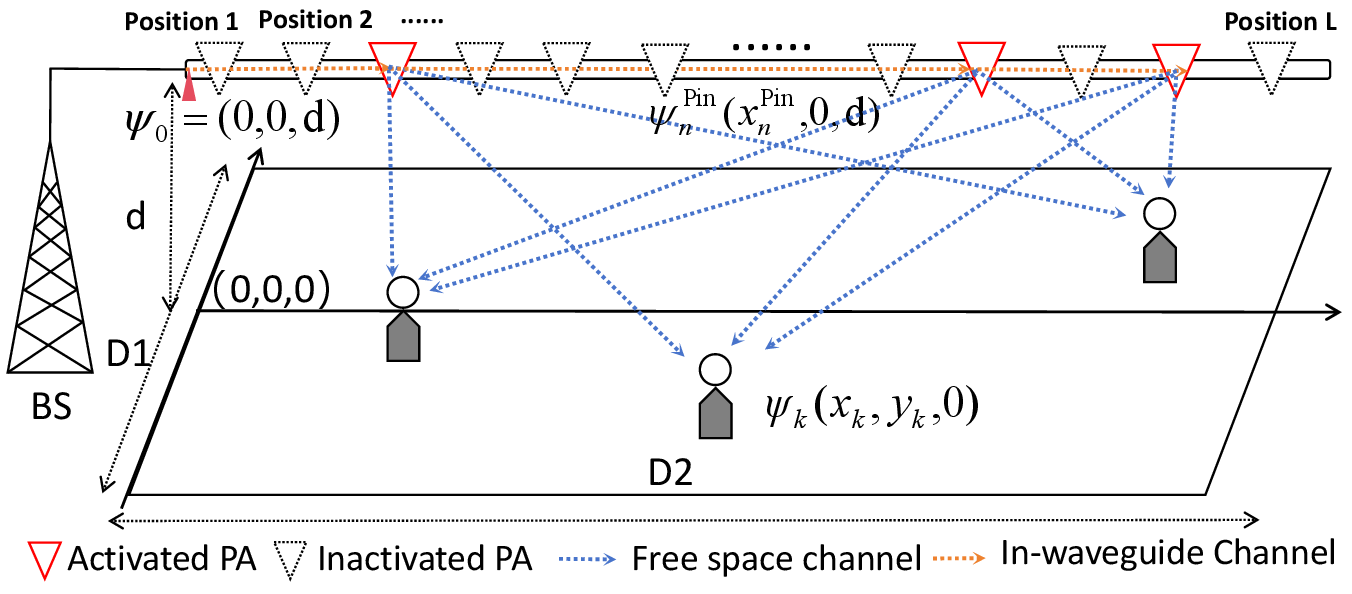}
\caption{System model for downlink NOMA-assisted discrete PASS.}
\label{system model}
\end{figure}

\vspace{-9pt}
\section{System Model And Problem Formulation}%
A rectangular region of width $D_1$ along the $y$-axis and length $D_2$ along the $x$-axis, as illustrated in Fig.~\ref{system model} is considered. $d$ is the height of the waveguide. The user $k$ is located at $(x_k, y_k, 0)$ with $K$ denoting the number of users. 
$N_\mathrm{m}$ PAs are pre-configured and uniformly placed along the waveguide,
the index set of equally-spaced PA positions is $\mathcal{L} = \{1, \cdots, l, \cdots, L\}$, with $L = N_\mathrm{m}$.
During transmission, only a subset of PAs is activated, whose label set is denoted by
$\mathcal{N}_{\mathrm{act}}$. The location of PA $n$ is represented by $\boldsymbol{\psi}_n^{\mathrm{Pin}} = (x_n, 0, d)$, $n \in \mathcal{N}_\mathrm{act}$. 

\subsection{NOMA PASS System Model and Power Consumption Model}%
In this system, users are ordered in the ascending order of channel gain to determine the successive interference cancellation (SIC) decoding order. Let $\boldsymbol{\pi}=[\pi_1,\ldots,\pi_k,\ldots,\pi_K]$ denotes a permutation of the original user indices, where $\pi_k$ is the user index at the $k$-th SIC decoding stage.
Let $g_\textit{k,n}=h_n \beta_{k,n}$ denotes the channel between the feed point and user $k$ through PA $n$ , where $h_n$ is the equivalent effective channel from feedpoint $\psi_0$ centred at $(0,0,d)$ to PA $n$ along the waveguide and
$\beta_{k,n}$ denotes the free-space propagation from PA $n$ to user $k$. Hence, 
\begin{equation}
g_{k,n}
=
\left(\!
\sqrt{10^{-\kappa L_n/10}}
\, e^{-j \tfrac{2\pi}{\lambda_g} L_n}
\!\right)\!
\left(\!
\frac{\lambda}{4\pi d_{k,n}}
\, e^{-j \tfrac{2\pi}{\lambda} d_{k,n}}
\!\right),
\label{eq:channel}
\end{equation}
where $d_{k,n}=|\psi_k-\psi_n^\text{Pin} |$ denotes the distance between user $k$ and PA $n$, $L_n=|\psi_0-\psi_n^\text{Pin} |$ represents the distance between PA $n$ to feedpoint $\psi_0$, $\lambda$ is the free-space wavelength, $\lambda_g = \lambda/n_{\mathrm{eff}}$ is the guided wavelength with $n_\mathrm{eff}$ as the effective refractive index and $\kappa$ is the attenuation factor (in dB/m) along the waveguide. At the user terminal, the received channel is a composite channel formed by the superposition of multiple links via the $N$ activated PAs. Therefore, the effective channel power gain for user $k$ is
$|g_k|^2 = \left| \sum_{n\in \mathcal{N}_\mathrm{act}} g_\textit{k,n}  \right|^2 =\left| \sum_{n\in \mathcal{N}_\mathrm{act}} h_n \beta_{k,n} \right|^2$.
Accordingly, the sorted users satisfy $|g_1|^2 \le |g_2|^2 \le \cdots \le |g_K|^2$.

The superimposed signal transmitted from the BS is $\sum_{k=1}^{K} \sqrt{\alpha_k}s_k$, where $s_k$ is the signal for user $k$ with unit energy. By assuming that the transmit power of
all PAs is equally distributed across the $N$ activated PAs \cite{PAP,Sumrate1,Matching}, the downlink signal received by user $k$ is given by
\begin{equation}
y_k = \sqrt 
{\frac{P_t\alpha_k}{N} } g_k s_k + \sum_{j=1,j \neq k}^{K}\sqrt {\frac{P_t\alpha_j}{N}}  g_k  s_j + n_k,
\end{equation}
where ${P_t}$ is the maximum transmit power budget, $n_k$ is the additive white Gaussian noise with power $\sigma^2 = N_0 B$, where $N_0$ is the noise power spectral density and $B$ is the total bandwidth of the system. 

Assuming perfect SIC and a decoding order determined by the users' effective channel gains, user $k$ cancels the signals of weaker users before decoding its own signal. Considering antenna activation, the achievable rate of user $k$ is given by
\begin{equation}
R_{k}(\mathcal{N}_\mathrm{act}) = B \log_2 \left( 1 + \frac{\frac{P_t}{N} |g_k|^2 \alpha_k }{\frac{P_t}{N} |g_k|^2 \sum_{j=k+1}^{K} \alpha_j + \sigma^2} \right).
\label{eq:rateNOMA}
\end{equation}
                                                        
The total power consumption of the considered PASS is 
\begin{equation}
    P_{\mathrm{tot}}
    = P_{\mathrm{static}}
    + \frac{P_t}{\eta}\sum\nolimits_{k=1}^K \alpha_k
    + N P_\text{act},
\end{equation}
where $\alpha_k$ is the power allocation factor for user $k$, $\eta$ is the power amplifier efficiency, and $P_{\mathrm{static}}$ is the static circuit power. Additionally, to characterize the power consumption in PASS more precisely, the power consumption for activating one is denoted as $P_\text{act}$, which scales linearly with the number of activated PAs $N=|\mathcal{N}_\mathrm{act}|$.

\vspace{-4mm}
\subsection{Problem Formulation}
The EE maximization problem in NOMA is formulated as
\begin{subequations}
\begin{align}
\max_{\mathcal{N}_\mathrm{act}, \boldsymbol{\pi}, \boldsymbol{\alpha}} \quad &
\frac{ \sum_{k=1}^{K} R_{k}(\mathcal{N}_\mathrm{act}) }
{P_{\text{tot}}}
\label{eq:EE_objNOMA}
\\
\text{s.t.}
\quad &
R_k(\mathcal{N}_\mathrm{act}) \ge R_k^{\min},  \forall k,
\label{eq:QoS_constraint}
\\\quad & 
\sum\nolimits_{k=1}^K \alpha_k \le 1, \alpha_k \ge 0, \forall k ,
\label{eq:sum_alpha}
\\\quad & 
\mathcal{N}_{\mathrm{act}}\subseteq\mathcal{L}, \;|\mathcal{N}_{\mathrm{act}}| \le N_m.
\label{eq:NNN}
\end{align}
\end{subequations}
where $\boldsymbol{\alpha} = [\alpha_1,\ldots,\alpha_K]$ is the power allocation coefficients vector. Constraint~\eqref{eq:QoS_constraint} guarantees the minimum rate of each user and \eqref{eq:sum_alpha} is the total power constraint. Constraint~\eqref{eq:NNN} guarantees that activated PAs are selected from the pre-configured PAs and their number cannot exceed the available PA count. The proposed two-layer algorithm for PA activation and power allocation is presented in the following sections.

\section{EE-Oriented Matching for PA Activation}
\label{sec:EE_matching}
\subsection{One-sided One-to-One Matching Model}
To employ the matching algorithm, we introduce $\mathcal{N}=\{1,\cdots,n,\cdots,N_\mathrm{m}\}$ as the label set of all PAs.
We then model the assignment between PAs label set $\mathcal{N}$ and positions $\mathcal{L}$ as a one-sided one-to-one matching $(\mathcal{N}, \mathcal{L}, \succ, \Phi_0)$,  
where $\succ$ denotes a list of strict preferences of the PAs over the positions and $\Phi_0$ is the initial matching which is randomly generated.

\begin{definition}[One-sided one-to-one matching \cite{Matching}]
\label{def:matching}
A matching $\Phi$ is a mapping $\Phi:\mathcal{N}\to\mathcal{L}\cup\{0\}$ such that
\begin{enumerate}[label=(\arabic*)] 
\item $\Phi(n)\in \mathcal{L}\cup\{0\},\ \forall n\in\mathcal{N};$ 
\item $|\Phi(n)|\in\{0,1\},\ \forall n\in\mathcal{N};$ and 
\item $\Phi(n)=l \Rightarrow \Phi(n')\neq l,\ \forall n,n'\in\mathcal{N}.$ 
\end{enumerate}
\end{definition}
Conditions in Definition~\eqref{def:matching} ensure that each PA is assigned to at most one position or remains unmatched, while each position is occupied by at most one PA. Specifically, $\Phi(n)=l$ indicates that the PA at position $l$ is labeled as PA $n$ and activated. In contrast, $\Phi(n)=0$ indicates that no PA is assigned to index $n$; hence $|\Phi|=N$. Additionally, we use the 2-tuple $\langle \Phi(n),\Phi\rangle$ to represent the assignment of PA $n$ under matching $\Phi$.


\vspace{-3mm}
\subsection{EE Utility Function}
With mapping $\Phi$ and the corresponding $\boldsymbol{\pi}$, the data rate of user $k$ is represented by $R_{k}(\Phi)$. The EE utility function is thus defined as
\begin{equation}
U(\Phi)\triangleq
\frac{\sum_{k=1}^{K}R_{k}\big(\Phi\big)}
{P_{\mathrm{static}}+\frac{P_t}{\eta}\sum_{k=1}^K \alpha_k^\star + |\Phi| P_{\mathrm{act}}},
\label{eq:utility-EE-TWC}
\end{equation}
where $\alpha_k^\star$ denotes the optimal power allocation factor which will be presented in Section~\uppercase\expandafter{\romannumeral4}.
In the case that PA $n \in \mathcal{N}$ prefers position $l \in \mathcal{L}$ over its current state, this strict preference can be characterized as
\begin{equation}
\left\langle \Phi(n),\,\Phi \right\rangle \prec_n \left\langle l,\,\Phi' \right\rangle \; \Leftrightarrow \; U(\Phi) < U(\Phi'),
\label{eq:preference}
\end{equation}
which means $\Phi'$ strictly prefers to $\Phi$ if $U(\Phi')>U(\Phi)$ \cite{Matching}.%

\vspace{-3mm}
\subsection{EE-Oriented Updates and Properties}
Since the channel gains vary with the matching $\Phi$, whenever $\Phi$ is updated, $\boldsymbol{\pi}$ is re-determined accordingly. The SIC order at iteration $i$ is represented by $\boldsymbol{\pi}_{i}$, the update of 
$\boldsymbol{\pi}$ is denoted as $\boldsymbol{\pi}_{i-1} \gets \boldsymbol{\pi}_{i}$, and $\boldsymbol{\pi}_{0}$ is obtained based on $\Phi_0$.

At convergence,  the final matching $\Phi^\star$ satisfies
\begin{equation}
\nexists\,\Phi' 
\quad\text{s.t.}\;
U(\Phi')>U(\Phi^\star),
\label{eq:EE_stability}
\end{equation}
where $\Phi'$ is generated from $\Phi^\star$ according to the matching update rule in Algorithm~\ref{alg:EE_matching}.
Once the algorithm converges to the final matching $\Phi^\star$, the activated PA label set used in the system model is obtained as $\mathcal{N}_{\mathrm{act}}=\{n\in\mathcal{N}\mid \Phi^\star(n)\neq 0\}$.
The matching algorithm is summarized in Algorithm~\ref{alg:EE_matching}.
\begin{algorithm}[t]\small
\caption{Outer Layer Matching Algorithm}
\label{alg:EE_matching}
\begin{algorithmic}[1]
\State \textbf{Initialization:} $\Phi_{0}$, $U(\Phi_{0})$, $\boldsymbol{\pi_{0}}$, $N_\mathrm{m}$, $L$ and set $i\gets 0$.
\State \textbf{Repeat}: $i \gets i+1$
    \For{$n = 1,2,\ldots,N_\mathrm{m}$}
        \For{each $l \in \mathcal{L}$}
            \If{$\Phi(n') \neq l,\ \forall\, n' \in \mathcal{N}$}
                \If{$\left\langle \Phi(n),\,\Phi \right\rangle \prec_{n}
                    \left\langle l,\,\Phi' \right\rangle$} \label{ln:cond1} 
                    \State Set $\Phi\gets\Phi'$. $\boldsymbol{\pi}_{i-1}\gets \boldsymbol{\pi}_{i}$, $U(\Phi_{i}) \gets U(\Phi')$. \label{up1}
                \EndIf
            \Else
                \If{$\Phi(n) = l$}
                    \If{$\left\langle \Phi(n),\,\Phi \right\rangle \prec_{n}
                        \left\langle 0,\,\Phi' \right\rangle$} \label{ln:cond2} 
                        \State Set $\Phi \gets \Phi'$. $\boldsymbol{\pi}_{i-1} \gets \boldsymbol{\pi}_{i}$, $U(\Phi_{i}) \gets U(\Phi')$. \label{up2}
                    \EndIf
                \EndIf
            \EndIf
        \EndFor
    \EndFor
\State \textbf{Until \eqref{eq:EE_stability} is satisfied.}
\State\Return $\boldsymbol{\pi}^{\star} \gets \boldsymbol{\pi}_{i}$ and $\Phi^{\star} \gets \Phi_{i}$
\end{algorithmic}
\end{algorithm}

\noindent\hspace*{1.2em}1) \textit{Complexity:} The outer matching procedure examines $\mathcal{O}(N_\mathrm{m}^2)$ candidate updates per iteration, each invoking the inner power-allocation procedure with complexity $\mathcal{C}_{\mathrm{in}}$. Hence, for $C$ outer iterations, the overall complexity is $\mathcal{O}\!\left(CN\mathrm{m}^2\mathcal{C}_{\mathrm{in}}\right)$.

\noindent\hspace*{1.2em}2) \textit{Convergence:} Each accepted matching update strictly increases the utility, while the number of feasible activation states is finite. Therefore, the proposed algorithm converges in a finite number of iterations to a state for which no further update can improve the utility.


\section{Power Allocation Optimization}
Since Algorithm~\ref{alg:EE_matching} requires repeated evaluations of $U(\Phi)$, the inner-layer power allocation must be solved efficiently. For each fixed $\Phi$ and $\boldsymbol{\pi}$, the power allocation subproblem of maximizing \eqref{eq:EE_objNOMA} subject to \eqref{eq:QoS_constraint} and \eqref{eq:sum_alpha} is a concave-convex fractional program that can be solved by a Dinkelbach-based algorithm with a closed-form update in each iteration.


\vspace{-5mm}
\subsection{Derivation for $\alpha_K^\star$}
\begin{proposition}
\label{thm:alpha_k}
If $\alpha_K$ is obtained, for any user $j,j \ne K$, an affine structure $\alpha_j = a_j \alpha_K + b_j$ can be utilized to represent $\alpha_j$, where $a_j$ and $b_j$ are constants.
\end{proposition}

\begin{proof}
Since $\Phi$ is the existing matching state, $N = |\Phi|$. The SINR for user $k$ to meet $R_k^\text{min}$ is denoted as $\gamma_k = 2^{\frac{R_k^{\min}}{B}} - 1$. Then $\alpha_k, \forall k$ can be represented by $\alpha_k \ge \gamma_k \sum_{j>k}^{K-1}\alpha_j + C_k$, 
where $C_k = \gamma_k \frac{\sigma^2 N}{P_t |g_k|^2}$.
This power-allocation problem is equivalent to a single-cluster downlink NOMA EE maximization problem. Following the optimality insight in \cite{Optimal power allocation}, the power coefficients of users $k\neq K$ should be adjusted only to maintain $R_k^{\min}$. 
Therefore, for user $k$, $\forall k \ne K$, we can obtain 
\begin{equation}
\alpha_k = \gamma_k \sum \nolimits_{j>k}^{K}\alpha_j + C_k, \forall k, k \ne K.
\label{eq:alpha_min_recursive}
\end{equation}
Hence, $\alpha_k$ admits an affine structure.
\end{proof}

\begin{proposition}
\label{thm:opt_pa}
Objective function~\eqref{eq:EE_objNOMA} can be transformed into a function of $\alpha_K$.
\end{proposition}
\begin{proof}

Define the aggregated expressions $A_{\mathrm{sum}} = \sum_{j=1}^K a_j$ and $B_{\mathrm{sum}} = \sum_{j=1}^K b_j$,
such that
\begin{equation}
\sum_{k=1}^K\!\alpha_k 
= \sum_{k=1}^{K-1} ( \gamma_k \sum_{j>k}\alpha_j + C_k)\!+\!\alpha_K=A_{\mathrm{sum}} \alpha_K\!+\!B_{\mathrm{sum}}.
\label{eq:Asum_Bsum}
\end{equation}
In the considered scenario, identical minimum data rate requirements are assumed for all users. Accordingly, the corresponding SINR target can be written as $\gamma_k=\gamma,\ \forall k$. Following the recursive expansion, we obtain $A_{\mathrm{sum}}= (1+\gamma)^{K-1}$ and $B_{\mathrm{sum}}= \sum_{i=1}^{K-1} C_{i}(1+\gamma)^{i-1}$. 

Building on Proposition ~\ref{thm:alpha_k}, we define the fractional objective function  as a ratio $\frac{f(\alpha_{K})}{g(\alpha_{K})}$. The numerator is
$f(\alpha_{K}) = (K-1) R_k^\text{min} + B \log_2(1+ \frac{P_t|g_K|^2}{N\sigma^2} \alpha_K)$.
Utilizing \eqref{eq:Asum_Bsum}, the denominator becomes
$g(\alpha_{K})
=
\frac{P_t}{\eta}(A_{\mathrm{sum}}\alpha_K + B_{\mathrm{sum}})
+ P_c
=
D_0 + D_1\alpha_K$,
with $P_c =P_\mathrm{static}+ N P_\mathrm{act}$, $D_0=\frac{P_t}{\eta}B_{\mathrm{sum}} + P_c$ and $D_1=\frac{P_t}{\eta}A_{\mathrm{sum}}$. Therefore, objective function is transformed into a function of $\alpha_K$. 
\end{proof}

With Proposition~\ref{thm:alpha_k} and \ref{thm:opt_pa}, the $i$th Dinkelbach iteration solves 
\begin{equation}
\max_{\alpha_{K,(i)}} \quad 
f(\alpha_{K,(i)})-\beta^{(i-1)}(D_0 + D_1 \alpha_{K,(i)}).
\label{eq:fbeta}
\end{equation}
According to \eqref{eq:fbeta}, the stationary point of $\alpha_K$ is 
\begin{equation}
\alpha_{K,(i)}^{\star}
=
\frac{B}{
\ln2\ \beta^{(i-1)} \frac{P_t}{\eta}(1+\gamma)^{K-1}
}
-
\frac{N\sigma^2}{P_t|g_K|^2}.
\label{eq:ak_closed}
\end{equation}

Using the global power constraint \eqref{eq:sum_alpha} together with
\eqref{eq:Asum_Bsum}, we obtain $A_{\mathrm{sum}}\alpha_K + B_{\mathrm{sum}} \le 1$. Thus, the maximum feasible $\alpha_K$ is $\alpha_K^{\max} =\min\!\left(1,\, \frac{1 - B_{\mathrm{sum}}}{A_{\mathrm{sum}}}\right)$, while $\alpha_K^{\min}$ can be derived from \eqref{eq:alpha_min_recursive} as $\alpha_K^{\min} = \gamma \frac{\sigma^2 N}{P_t |g_K|^2}$. Hence, the feasible interval of $\alpha_K$ is $[\alpha_K^{\min},\ \alpha_K^{\max}]$. 
Since \eqref{eq:fbeta} is strictly concave on the interval $[\alpha_K^{\min},\alpha_K^{\max}]$, there exists a unique global maximum. Finally, the optimal solution at the $i$-th iteration is
\begin{equation}
\alpha_{K,(i)}^{\star} =\min\{\alpha_K^{\max},\max\{\alpha_K^{\min},\alpha_{K,(i)}^{\star}\}\}.
\label{alphakstar}
\end{equation}

\vspace{-3mm}
\subsection{Dinkelbach-based Inner-layer Power Allocation Procedure}

Based on the above derivation, the inner-layer power allocation is solved by a Dinkelbach-based iterative procedure. 
The initialization $\beta^{0}$ is chosen as the EE corresponding to an equal-power allocation $\boldsymbol{\alpha}^{(0)}$, where $\alpha_{k,(0)} = \frac{1}{K},\ \forall k$, and
\begin{equation}
    \beta^{(0)}
    = \frac{\sum_{k=1}^K
        B \log_2(1 + \frac{\frac{P_t}{N} |g_k|^2 \alpha_{k, (0)}}{\frac{P_t}{N}|g_k|^2\sum_{j>k}\alpha_{j, (0)} + \sigma^2})}{
        P_{\mathrm{static}}
        + \frac{P_t}{\eta}\sum_{k=1}^K \alpha_{k, (0)}
        + NP_{\mathrm{act}}
    }.
\label{beta0NOMA}
\end{equation}
At the $i$-th Dinkelbach iteration, $\alpha_{K, (i)}^{\star}$ is firstly calculated according to \eqref{eq:ak_closed} and \eqref{alphakstar}, and the power coefficients of users $k\neq K$ are then calculated according to \eqref{eq:alpha_min_recursive}. The Dinkelbach parameter in iteration $i$ is updated by
$\beta^{(i)}
=
\frac{
f(\alpha_{K, (i)}^{\star})
}{
g(\alpha_{K, (i)}^{\star}).
}$
The iteration terminates once
$\left|
\beta^{(i)} - 
\beta^{(i-1)} 
\right|
\le \varepsilon,
\label{eq:dinkelbach_stop}$
where $\varepsilon$ is a small threshold (e.g., $10^{-6}$).

\vspace{-3mm}
\subsection{Analytical Impact of $P_{\mathrm{act}}$}
\label{Mathematical_Pact}
Let $\mathbf{x}$ denote a tuple containing feasible transmission configurations $(\mathcal{N}_{\mathrm{act}},\boldsymbol{\alpha},\boldsymbol{\pi})$ of the system with its sum rate, transmit power, and number of activated PAs as $R(\mathbf{x})$, $P(\mathbf{x})$, and $N(\mathbf{x})$, respectively.
For a fixed $\mathbf{x}$, the inverse EE is
$J(\mathbf{x},P_{\mathrm{act}})
\triangleq\frac{1}{\mathrm{EE}(\mathbf{x},P_{\mathrm{act}})}
=\mu(\mathbf{x})+\nu(\mathbf{x})P_{\mathrm{act}}$,
where $\mu(\mathbf{x})\triangleq\frac{P_{\mathrm{static}}+P(\mathbf{x})/\eta}{R(\mathbf{x})}$,
$\nu(\mathbf{x})\triangleq\frac{N(\mathbf{x})}{R(\mathbf{x})}$.
Thus, $\nu(\mathbf{x})$ characterizes the sensitivity of the inverse EE to $P_{\mathrm{act}}$.

\begin{proposition}
\label{prop:Pact_sensitivity}
For $P_{\mathrm{act},i} < P_{\mathrm{act},j}$, let $\mathbf{x}_i^\star$ and $\mathbf{x}_j^\star$ denote globally EE-optimal configurations under $P_{\mathrm{act},i}$ and $P_{\mathrm{act},j}$, respectively. Then
\begin{equation}
{N(\mathbf{x}_i^\star)}/{R(\mathbf{x}_i^\star)}
\geq
{N(\mathbf{x}_j^\star)}/{R(\mathbf{x}_j^\star)}.
\label{eq:NR_monotonic}
\end{equation}
\end{proposition}

\begin{proof}
Let $\mu_l=\mu(\mathbf{x}_l^\star)$ and $\nu_l=\nu(\mathbf{x}_l^\star)$, $l\in\{i,j\}$. By optimality,
$\mu_i+\nu_iP_{\mathrm{act},i}\leq\mu_j+\nu_jP_{\mathrm{act},i}$,
$\mu_j+\nu_jP_{\mathrm{act},j}\leq\mu_i+\nu_iP_{\mathrm{act},j}$.
Hence, $(\nu_i-\nu_j)(P_{\mathrm{act},i}-P_{\mathrm{act},j})\leq0$.
Since $P_{\mathrm{act},j}>P_{\mathrm{act},i}$, we obtain $\nu_j\leq\nu_i$, which proves \eqref{eq:NR_monotonic}.
\end{proof}

For a given $\mathbf{x}_j^\star$ and since $\nu_l>0$,
$J(\mathbf{x}_j^\star,P_{\mathrm{act},j})
>
J(\mathbf{x}_j^\star,P_{\mathrm{act},i})
\geq
J(\mathbf{x}_i^\star,P_{\mathrm{act},i})$,
where the second inequality follows from the optimality of $\mathbf{x}_i^\star$ at $P_{\mathrm{act},i}$. Therefore,
$\mathrm{EE}^\star(P_{\mathrm{act},j})<\mathrm{EE}^\star(P_{\mathrm{act},i})$.
Hence, although a larger $P_{\mathrm{act}}$ may alter the optimal configuration $\mathbf{x}$, it cannot offset the resulting EE loss.

To further investigate the impact of $P_{\mathrm{act}}$ on  power allocation, the two-user case is considered with $|g_1|^2 \leq |g_2|^2$.
According to \eqref{eq:alpha_min_recursive}, the sum of power allocation factors equals $\alpha_1+\alpha_2=(1+\gamma)\alpha_2+C_1$.

When the transmit power budget is stringent, the EE solution will utilize all the available power. Therefore $\alpha_1^{\star}+\alpha_2^{\star}=1$, with $\alpha_1^{\star}=\frac{\gamma+C_1}{1+\gamma}$ and $\alpha_2^{\star}=\frac{1-C_1}{1+\gamma}$.
Since neither $\alpha_1^{\star}$ nor $\alpha_2^{\star}$ contains $P_{\mathrm{act}}$, the power allocation factors are independent of $P_{\mathrm{act}}$ when $P_t$ is small. 
In contrast, when the transmit power budget is larger than necessary, $\alpha_1^{\star}+\alpha_2^{\star} < 1$, and the power allocation factors are determined by the Dinkelbach algorithm. At the $(i-1)$-th
 iteration, the parameter $\beta^{(i-1)}$ becomes
\begin{equation}
\beta^{(i-1)}
=
\frac{
R_1^{\min}
\!+\!
B\log_2\!
\left(
1+
\frac{P_t|g_2|^2}
{N\sigma^2}
\alpha_{2, (i-1)}
\right)
}{
\frac{P_t}{\eta}
\left[
(1\!+\!\gamma)\alpha_{2, (i-1)}\!+\!C_1
\right]
\!+\!
P_{\mathrm{static}}
\!+\!
NP_{\mathrm{act}}
}.
\label{eq:loose_beta}
\end{equation}
For a given $\alpha_{2, (i-1)}$, increasing $P_{\mathrm{act}}$
increases the denominator of \eqref{eq:loose_beta},  thereby reducing $\beta^{(i-1)}$.
According to \eqref{eq:ak_closed},
${\alpha}_{2,(i)}$ is inversely related to
$\beta^{(i-1)}$ and thus, decreasing $\beta^{(i-1)}$ increases ${\alpha}_{2,(i)}^{\star}$, which linearly increases $\alpha_{1,(i)}^{\star}$ based on \eqref{eq:alpha_min_recursive}.
This shows in the case of ample transmit power budget, all power allocation factors increases with increasing $P_{\mathrm{act}}$. In other words when there are additional power available, the EE solution will increase the transmit power to achieve a higher rate in order to compensate the increased activation power.

\vspace{-3mm}
\section{Simulation results}

\begin{figure}[!t]
\centering
\includegraphics[width=0.75\columnwidth]{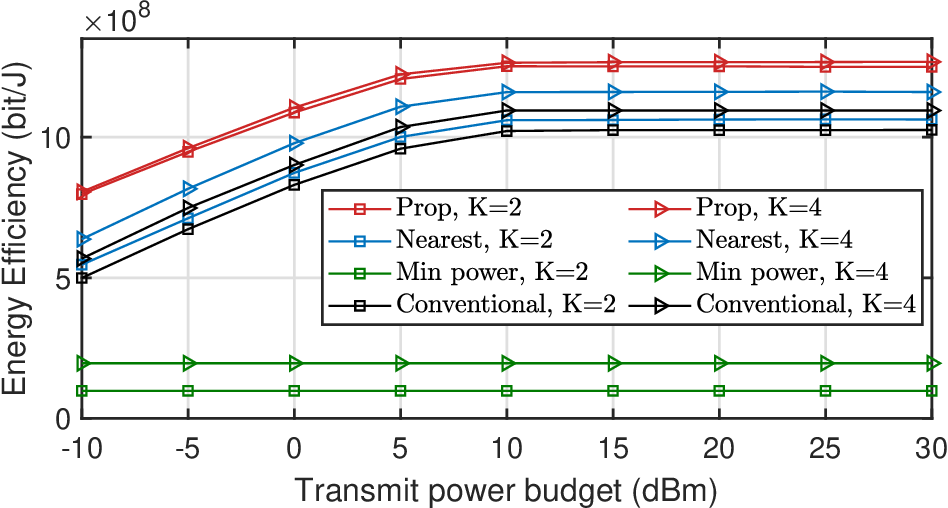}
\caption{EE performance for the proposed algorithm and benchmarks.}
\label{Baselines}
\end{figure}

This section presents the simulation results of the proposed joint PA selection and power allocation algorithm (denoted as ``Prop"). We set $D_1= 10\mathrm{m}$, $D_2= 20\mathrm{m}$, the length of waveguide is equal to $D_2$, $d= 3\mathrm{m}$, $L = 40$, the total bandwidth of the system is $B=10~\mathrm{MHz}$, the carrier frequency $f_c=28~\mathrm{GHz}$, $N_0=-174~\mathrm{dBm/Hz}$, $n_\mathrm{eff} = 1.4$, $\kappa = 0.02~\mathrm{dB/m}$, $P_{\mathrm{static}} = 0.1~\mathrm{W}$ and $P_\text{act} = 3~\mathrm{dBm}$ unless otherwise stated. 
For comparison, several benchmark schemes employing NOMA are considered:
(1) Conventional antennas located at $\psi_n^\text{Con} = \psi_0$ with optimal power allocation (denoted as ``Conventional").
(2) Minimum transmit power satisfying the rate constraints (denoted as ``Min power") \cite{EEminimization0}.
(3) Selecting the nearest PAs for each users with optimal power allocation (denoted as ``Nearest").
(4) Exhaustively searching over all feasible activated PAs sets with optimal power allocation (denoted as ``Exhaustive").

Fig.~\ref{Baselines} compares the EE performance of the proposed scheme with the benchmarks against $P_t$ when $K$ equals 2 and 4. It is observed that the EE monotonically increases when $P_t \le 10$ dBm, and then saturates, as further increasing $P_t$ achieves marginal rate gain comparing to the linear increase in total energy consumption. The proposed algorithm also outperforms conventional antennas even though the latter does not have the PA activation power. This highlights the clear advantage of the proposed NOMA PASS over the conventional setup.
Additionally, the EE performance of the considered system increases with $K$.
The proposed algorithm also achieves higher EE than the ``Nearest" scheme, confirming the benefit of the matching design and the importance of activated PAs selection optimization.
Moreover, it is also clear that the minimum power strategy performs poorly in EE as it does not balance between rate and power consumption.

\begin{figure}[!t]
\centering
\subfigure{%
    \includegraphics[width=0.31\linewidth,height=41.5mm]{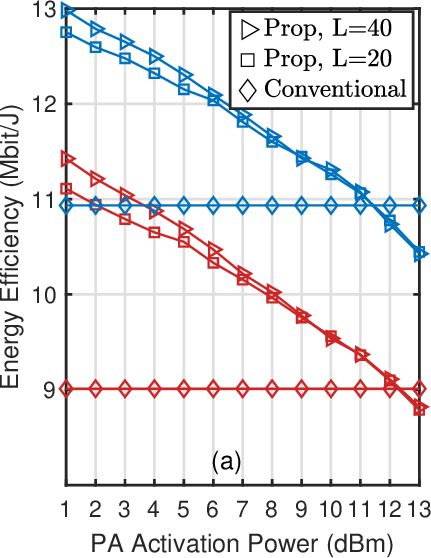}%
}%
\hspace{0.01\linewidth}%
\subfigure{%
    \includegraphics[width=0.31\linewidth,height=41.5mm]{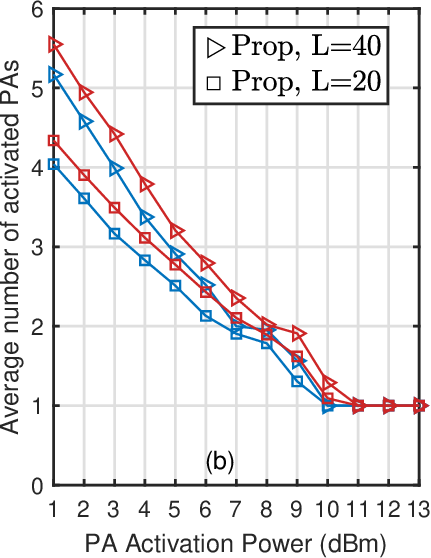}%
}%
\hspace{0.01\linewidth}%
\subfigure{%
    \includegraphics[width=0.36\linewidth,height=41.5mm]{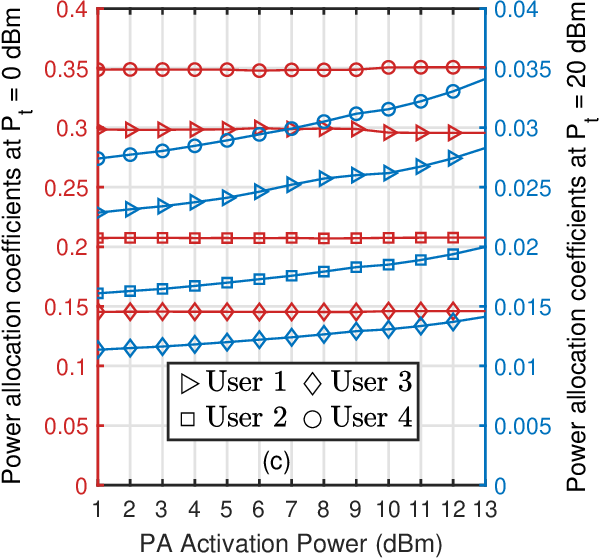}%
}%
\vspace{-1mm}
\caption{System performance versus $P_{\mathrm{act}}$ under different
$P_t$ with $K=4$: (a) EE, (b) $N$, and (c) $\boldsymbol{\alpha}$ in the ``Prop'' scheme. The red and blue curves correspond to $P_t=0$ dBm and $P_t=20$ dBm, respectively.}
\label{NOMA system}
\end{figure}
Fig.~\ref{NOMA system}(a) shows the EE performance of the proposed algorithm versus $P_t$ for different $P_\text{act}$ values and PA deployment densities.  The EE decreases monotonically as $P_\text{act}$ increases, confirming the analytical decreasing results in Section~\ref{Mathematical_Pact}, and a crossover point appears beyond which the proposed PASS becomes less energy-efficient than the conventional antenna system, highlighting the importance of power-efficient PA activation mechanism. In addition, PA density has a stronger impact on EE at low $P_\text{act}$, while this impact weakens at high $P_\text{act}$ because activation power dominates the total consumption and reduces the average number of activated PAs. As shown in Fig.~\ref{NOMA system}(b), the number of activated PAs gradually decreases with increasing $P_\text{act}$ and eventually reduces to one, implying that the activation cost of additional PAs outweighs their marginal rate gain. It can be observed that all schemes exhibit a slight plateau when the average number of activated PAs is close to 2. This occurs because, in this region, the feasible activation options are limited to 1 to 3 with no smaller configuration available; hence, although the activation ratio does vary, the resulting average remains clustered around 2. In addition, the number of activated PAs is larger at  $P_t$ = 0 dBm than 20 dBm, as more PAs are needed to achieve the rate requirements when the power budget is small, which also contributes to a lower EE.
Fig.~\ref{NOMA system}(c) demonstrates the impact of $P_\text{act}$ on power allocation. When $P_\mathrm{act}$ increases, the power allocation coefficients remain unchanged under stringent power budget at $P_t $ = 0 $\mathrm{dBm}$, while gradually increase when $P_t$ is larger than necessary at 20 $\mathrm{dBm}$, but with smaller values for lower transmission power. This behavior is also consistent with the the analysis presented in Section~\uppercase\expandafter{\romannumeral4}-C. 
Nevertheless, Fig 3(a) shows that when the activation power is high, the achievable EE becomes nearly identical in both cases, indicating that $P_\mathrm{act}$ dominates the power consumption and substantially diminish the potential EE gain even when additional transmit power budget is available.


Fig.~\ref{NOMA convergence} shows the convergence behavior of the proposed algorithm. In Fig.~\ref{NOMA convergence}(a), EE increases monotonically and approaches the exhaustive search result. A larger number of candidate PAs results in slightly higher EE but slower convergence due to the enlarged search space. In Fig.~\ref{NOMA convergence}(b), the average number of activated PAs gradually decreases and stabilizes near the exhaustive search result, showing that redundant activated PAs can be effectively removed while preserving near-optimal EE. The results further indicate that a one-to-one mapping between users and PAs is not required, since a smaller PA set can optimally serve multiple users with lower power consumption.

\section{Conclusions}

In this paper, we investigated an EE maximization problem with PA activation power consumption for a NOMA-assisted downlink PASS. 
A two-layer optimization algorithm is proposed with a matching-based PA activation procedure, and the Dinkelbach-based optimal power allocation. Numerical results demonstrate that the proposed algorithm converges rapidly and achieves EE performance close to the exhaustive-search optimal results and consistently outperforms other baselines. It further illustrates the importance of including PA activation power in EE maximization, and the need for low PA activation mechanism for an efficient PASS.

\begin{figure}[!t]
\centering
\subfigure{
    \includegraphics[width=0.64\columnwidth]{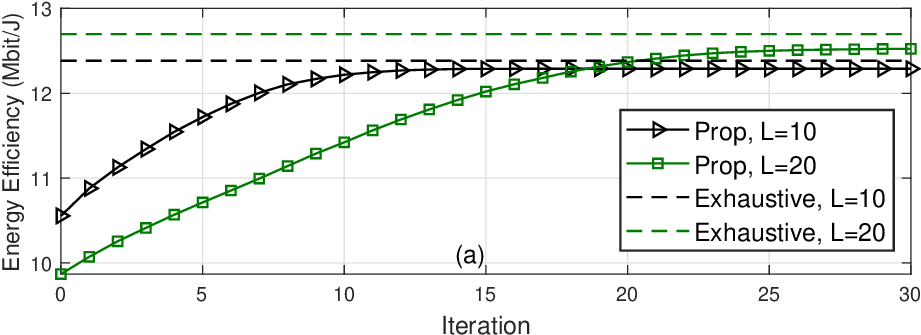}
}%
\hfill
\subfigure{
    \includegraphics[width=0.64\columnwidth]{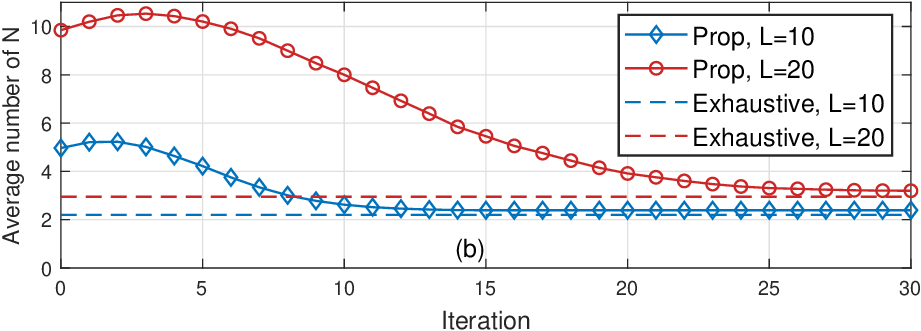}
}%
\vspace{-3pt}
\caption{Convergence performance of the proposed algorithm with
$P_t = 20$ dBm, $K = 5$, and $P_\text{act} = 3$ dBm.}
\label{NOMA convergence}
\end{figure}%



\newpage

\vfill


\begin{thebibliography}{1}
\bibliographystyle{IEEEtran}



\bibitem{Tutorial}
Y. Liu et al., ``Pinching-antenna systems (PASS): A tutorial," \textit{IEEE
Trans. Commun.}, vol. 74, pp. 4881–4918, 2026.

\bibitem{PAP}
Z. Ding, R. Schober and H. Vincent Poor, ``Flexible-antenna systems: A pinching-antenna perspective," \textit{IEEE Transactions on Communications}, vol. 73, no. 10, pp. 9236-9253, Oct. 2025

\bibitem{Sumrate1}
Z. Zhou, Z. Yang, G. Chen, and Z. Ding, ``Sum-rate maximization
for noma-assisted pinching-antenna systems," \textit{IEEE Wireless Commun.
Lett.}, vol. 14, no. 9, pp. 2728–2732, Sep. 2025.

\bibitem{Matching}
K. Wang, Z. Ding, and R. Schober, ``Antenna activation for NOMA-assisted pinching-antenna systems,'' \emph{IEEE Wireless Commun. Lett.}, vol. 14, no. 5, pp. 1526--1530, May 2025.

\bibitem{RSMA}
P. Wang, X. He, H. Wang, and R. Song, “Sum rate maximization for pinching antennas assisted RSMA system with multiple waveguides,” \emph{IEEE Trans. Veh. Technol.}., Early Access, Dec. 2025.

\bibitem{EEmax0}
M. Zeng, J. Wang, G. Zhou, F. Fang, and X. Wang, ``Energy-efficient design for downlink pinching-antenna systems with QoS guarantee,'' \emph{IEEE Trans. Veh. Technol.}, Early Access, Aug. 13, 2025.

\bibitem{EEuplink}
M. Zeng, X. Li, J. Wang, G. Huang, O. A. Dobre, and Z. Ding, ``Energy-efficient resource allocation for NOMA-assisted uplink pinching-antenna systems,'' \emph{IEEE Wireless Commun. Lett.}, vol. 14, no. 11, pp. 3695--3699, Nov. 2025.


\bibitem{EEminimization0}
S. Mohammadzadeh, K. Cumanan, C. Li, and Z. Ding, ``Efficient downlink power allocation for NOMA-based pinching-antenna systems,'' \emph{IEEE Wireless Commun. Lett.}, vol. 14, no. 12, pp. 4187--4191, Dec. 2025.


\bibitem{EEminBF0}
Z. Wang, C. Ouyang, X. Mu, Y. Liu, and Z. Ding, ``Modeling and beamforming optimization for pinching-antenna systems,'' \emph{IEEE Trans. Commun.}, vol. 73, no. 12, pp. 13904--13919, Dec. 2025.

\bibitem{EEminBF1}
D. Gan, X. Xu, J. Zuo, X. Ge, and Y. Liu, ``Joint beamforming for NOMA assisted pinching antenna systems (PASS),'' \emph{IEEE Trans. Commun.}, vol. 74, pp. 2450--2465, 2026.

\bibitem{JRBOM}
Y. Xu, D. Xu, X. Yu, S. Song, Z. Ding, and R. Schober, ``Joint radiation power, antenna position, and beamforming optimization for pinching-antenna systems with motion power consumption,'' \emph{IEEE Trans. Wireless Commun.}, vol. 25, pp. 7825--7841, 2026.

\bibitem{Fu_GLOBECOM}
Y. Fu, F. He, Z. Shi, and H. Zhang, ``Power minimization for NOMA-assisted pinching antenna systems with multiple waveguides,'' in \emph{Proc. IEEE Global Commun. Conf. (GLOBECOM)}, Taipei, Taiwan, 2025, pp. 170--175.

\bibitem{SEEE-tradeoff0}
Z. Zhou, Z. Wang, and Y. Liu, ``Spectral and energy efficiency tradeoff for pinching-antenna systems,'' \emph{arXiv preprint} arXiv:2510.25192, 2025.

\bibitem{SEEE-tradeoff1}
G. Zhu, X. Mu, L. Guo, \emph{et al.}, ``SE--EE tradeoff in pinching-antenna systems: Waveguide multiplexing or waveguide switching?,'' \emph{arXiv preprint} arXiv:2601.04844, 2026.


\bibitem{IncludePact}
X. Gan, Z. Wang, and Y. Liu, ``Dual-scale antenna deployment for pinching antenna systems,'' \emph{arXiv preprint} arXiv:2510.27185, 2025.

\bibitem{Optimal power allocation}
S. Rezvani, E. A. Jorswieck, R. Joda, and H. Yanikomeroglu, ``Optimal power allocation in downlink multicarrier NOMA systems: Theory and fast algorithms,'' \emph{IEEE J. Sel. Areas Commun.}, vol. 40, no. 4, pp. 1162--1189, Apr. 2022.




\end{thebibliography}
\end{document}